\documentclass[a4paper,USenglish,cleveref, autoref, thm-restate]{lipics-v2021}

\usepackage{todonotes}
\usepackage{algpseudocode}
\usepackage{algorithm}
\hideLIPIcs
\nolinenumbers
\newcommand{\cc}{\textsc{CongestedClique}}
\newcommand{\congest}{\textsc{Congest}}
\newcommand{\local}{\textsc{Local}}
\newtheorem*{tempTheorem}{Theorem}

\title{Congested Clique Counting for Local Gibbs Distributions}

\author{Joshua Z. Sobel}{Department of Computer Science, University of Iowa, Iowa City, Iowa, USA}{joshua-sobel@uiowa.edu}{https://orcid.org/0009-0004-7482-0754}{}
\authorrunning{J.\,Z. Sobel} 

\Copyright{Joshua Z. Sobel} 



\funding{The author was partially supported during this project by National Science Foundation grants 1955939 and 2402835.}

\acknowledgements{Thanks to Sriram Pemmaraju for suggestions and discussions.}

\category{} 

\relatedversion{} 

\EventEditors{John Q. Open and Joan R. Access}
\EventNoEds{2}
\EventLongTitle{42nd Conference on Very Important Topics (CVIT 2016)}
\EventShortTitle{CVIT 2016}
\EventAcronym{CVIT}
\EventYear{2016}
\EventDate{December 24--27, 2016}
\EventLocation{Little Whinging, United Kingdom}
\EventLogo{}
\SeriesVolume{42}
\ArticleNo{23}

\begin{document}

\maketitle

\begin{abstract}
There are well established reductions between combinatorial sampling and counting problems (Jerrum, Valiant, Vazirani TCS 1986).  Building off of a very recent parallel algorithm utilizing this connection (Liu, Yin, Zhang arxiv 2024), we demonstrate the first approximate counting algorithm in the \cc{} for a wide range of problems.  Most interestingly, we present an algorithm for approximating the number of $q$-colorings of a graph within $\epsilon$-multiplicative error, when $q>\alpha\Delta$ for any constant $\alpha>2$, in $\Tilde{O}\big(\frac{n^{1/3}}{\epsilon^2}\big)$ rounds.  More generally, we achieve a runtime of $\Tilde{O}\big(\frac{n^{1/3}}{\epsilon^2}\big)$ rounds for approximating the partition function of Gibbs distributions defined over graphs when simple locality and fast mixing conditions hold.  Gibbs distributions are widely used in fields such as machine learning and statistical physics.  We obtain our result by providing an algorithm to draw $n$ random samples from a distributed Markov chain in parallel, using similar ideas to triangle counting (Dolev, Lenzen, Peled DISC 2012) and semiring matrix multiplication (Censor-Hillel, Kaski, Korhonen, Lenzen, Paz, Suomela PODC 2015).  Aside from counting problems, this result may be interesting for other applications requiring a large number of samples.  In the special case of estimating the partition function of the hardcore model, also known as counting weighted independent sets, we can do even better and achieve an $\Tilde{O}\big(\frac{1}{\epsilon^2}\big)$ round algorithm, when the fugacity $\lambda \leq \frac{\alpha}{\Delta-1}$, where $\alpha$ is an arbitrary constant less than $1$.
\end{abstract}

\section{Introduction}
In a pioneering work, Valiant proposed the computational complexity class \#P \cite{ValiantSharpP}.  Problems in \#P are the counting variants of decision problems in NP; for example, counting the number of satisfying 3-SAT assignments or counting the number of proper vertex colorings of a graph using at most $q$ colors ($q$-colorings).  Unfortunately, Valiant established that even when a decision problem is in $P$, such as determining whether a bipartite graph has a perfect matching, the corresponding counting problem can be NP-hard.  Thus for most counting problems, we need to settle for $\epsilon$ multiplicative approximation rather than exact counting.  The goal is to estimate a true count $X$ by $\hat{X}$, such that $(1-\epsilon)X\leq\hat{X}\leq (1+\epsilon)X$ with probability at least $\frac{3}{4}$.\footnote{By the standard median trick, $3/4$ can be boosted to any constant less than one.}  

The typical method for approximate counting involves drawing random samples that are solutions to the corresponding search problem.  In other words, we can count the number of colorings of a graph by sampling random colorings.  This approach dates back to Jerrum, Valiant, and Vazirani \cite{JVV} showing that for a wide class of problems, most notably joint distributions on a graph (e.g. colorings and matchings), fully polynomial-time randomized approximate counting is inter-reducible with polynomial-time approximate sampling.\footnote{The precise result is slightly more complex and requires sampling from both the original distribution and derived conditional distributions.}  The initial reduction from counting to sampling \cite{JVV} requires drawing a very large number of samples; this has been improved over several subsequent works \cite{EdgeRemoveCounting, SimAnneal1, SimAnneal2, ParGibCount}.  In particular, when considering Gibbs distributions defined over graphs with $n$ vertices and meeting mild conditions, Štefankovič, Vempala, and Vigoda \cite{SimAnneal2} showed that counting can be accomplished by taking $\Tilde{O}\big(\frac{n}{\epsilon^2}\big)$ samples.\footnote{As we will explain later, for an arbitrary Gibbs distribution, counting refers to estimating the partition function.  This is a generalization of counting the size of a set.}  Gibbs distributions originated in physics, and generalize many problems such as the uniform distribution of $q$-colorings.  In computer science, they have found uses in many areas such as machine learning \cite{MachineLearning}, image processing \cite{ImageProcessing}, and even computational biology \cite{ComputationalBiology}.

Unfortunately, the algorithm from \cite{SimAnneal2} cannot be fully parallelized, as samples drawn later depend on samples drawn earlier.  It does seem possible to obtain an $\Tilde{O}(\frac{\sqrt{n}}{\epsilon^2})$ round \cc{} implementation of this algorithm; however, our results show we can do much better.  Very recently, Liu, Yin, and Zhang \cite{ParGibCount} also showed that for many Gibbs distributions defined over graphs counting can be done using $\Tilde{O}\big(\frac{n}{\epsilon^2}\big)$ total samples, but in their approach every sample can be taken in parallel.  This implies that counting in the \cc{} reduces to quickly drawing $\Tilde{O}\big(\frac{n}{\epsilon^2}\big)$ samples.  From existing work, \cite{WhatCanBeSampled,FischerGhaffari,FengHayesYin}, we know that for many Gibbs distributions drawing a single sample within total variation distance error of $\delta$ can be done in $O(\log \frac{n}{\delta})$ \cc{} rounds, using the \textit{distributed Metropolis-Hastings} Markov chain.  For Gibbs distributions that we call \textit{local}, we identify a novel structural similarity between simulating one transition of $n$ instances of the Markov chain and semiring matrix multiplication.  In particular, utilizing the approach of Dolev, Lenzen, and Peled \cite{TriangleCounting} for triangle counting and Censor-Hillel, Kaski, Korhonen, Lenzen, Paz, and Suomela \cite{MatrixMult} for semiring matrix multiplication, we can simulate one transition of $n$ instances of the Markov chain in $O(n^{1/3})$ rounds.  Altogether, this shows that for local Gibbs distributions we can count in $\Tilde{O}\big(\frac{n^{1/3}}{\epsilon^2}\big)$ \cc{} rounds when the Markov chain has an $O(\log\frac{n}{\delta})$ mixing time.

\subsection{Main Results}
Our central contribution is an algorithm to draw many samples at once from the distributed Metropolis-Hastings Markov chain.  Beyond counting, this result may be useful in other algorithms where several samples are needed.

\begin{restatable}{theorem}{mainresult}\label{simulationTheorem}One transition can be simulated for $n$ instances of the distributed Metropolis-Hastings chain in $O(n^{1/3})$ rounds.
\end{restatable}

As a corollary to the previous theorem and the results of \cite{ParGibCount}, under mild locality and mixing conditions which we will describe later, there is an $\Tilde{O}\big(\frac{n^{1/3}}{\epsilon^2}\big)$ round counting algorithm.  As the most interesting example, we obtain a fast algorithm for approximating the number of q-colorings of an input graph with degree $\Delta$.

\begin{restatable}{corollary}{colorcount}
    There is an algorithm for approximating the number of $q$-colorings of a graph in $\Tilde{O}(\frac{n^{1/3}}{\epsilon^2})$ rounds, when $q > \alpha\Delta$, for $\alpha>2$, within multiplicative error $\epsilon$, with probability at least $\frac{3}{4}$.
\end{restatable}

In the special case of the hardcore model, we can do better.  The hardcore model is a distribution over the independent sets of a graph, taking a parameter $\lambda$ called the fugacity.  The probability of independent set $I$ is proportional to $\lambda^{|I|}$.  For $\lambda \leq \frac{\alpha}{\Delta-1}$, $\alpha < 1$, we can draw $n$ samples in $\Tilde{O}(\frac{1}{\epsilon^2})$ rounds.  Counting in this model refers to estimating the partition function, $\sum_I \lambda^{|I|}$.

\begin{restatable}{corollary}{hardcorecounting}
    We can approximate the partition function of the hardcore model when $\lambda \leq \frac{\alpha}{\Delta-1}$ for $\alpha < 1$ in $\Tilde{O}(\frac{1}{\epsilon^2})$ rounds, within multiplicative error $\epsilon$ with probability at least $\frac{3}{4}$.
\end{restatable}

Many problems in the \cc{} can be solved very fast, often in $O(1)$ rounds.  For example, sorting \cite{LenzenRouting}, MST \cite{fastMST}, and $\Delta+1$-coloring \cite{fastColoring}.  Therefore, it is natural to wonder if counting problems could also be solved much faster than $\Tilde{O}(\frac{n^{1/3}}{\epsilon^2})$ rounds.  One problem for which no $\Tilde{O}(1)$ round algorithm is known, is detecting whether the input graph contains a triangle.  The current best known algorithm for triangle detection takes $O(n^{0.157})$ rounds, using distributed matrix multiplication \cite{MatrixMult} and the state of the art for sequential matrix multiplication by Williams, Xu, Xu, and Zhou \cite{NewMatrixMult}.  This uses the fact that the trace of the adjacency matrix cubed is nonzero if and only if the graph has a triangle.  This improves on an earlier result \cite{TriangleCounting} taking $O(n^{1/3})$ rounds. 

Directly counting the number of triangles in the input graph does not fall into the category of counting problems our approach can handle; therefore it does not offer us a direct lower bound.  However, when $\epsilon \leq \frac{1}{32n^2}$, we can use a different counting problem to detect whether the graph contains a triangle.  This gives us a conditional lower bound for counting in the class of problems we consider.  The precise lower bound is given in \cref{lowerBound}.
\begin{tempTheorem}[Lower Bound]Approximating the partition function of a general local Gibbs distribution when $\epsilon \leq \frac{1}{32n^2}$ with success probability $p$ is as hard as triangle detection with success probability $p$.  This holds even when the locality and fast mixing conditions are met.
\end{tempTheorem}

\subsection{Related Work}
There have been several recent results on distributed counting and sampling from complex distributions \cite{WhatCanBeSampled, FischerGhaffari, FengHayesYin, ExactDistSampl, distJVV, simMetropolis, LLLsampling, distFlipDynamics, distSpanningTree, ccSpanningTree}.  In addition to the results discussed earlier, a few more have bearing on our results here.  First, the only other result on sampling/counting these types of complex combinatorial structures in the \cc{} model is \cite{ccSpanningTree}; however, this result focuses on sampling random spanning trees and not counting.  In the \local{} model, Feng and Yin \cite{distJVV} gave a distributed version of the reduction from exact sampling to approximate inference (a distributed variant of counting).  This result is in the opposite direction of our result from approximate counting to approximate sampling.  This result is also not amenable to the \cc{} model as each vertex needs a large amount of local information.  Finally, in the \local{} model a few lower bounds are known.  First, Guo, Jerrum, and Liu \cite{LLLsampling} showed that sampling from the hardcore model requires $\Omega(\log n)$ rounds.  This same lower bound was established by Feng, Sun, and Yin \cite{WhatCanBeSampled} for $q$-colorings.  It is unknown if faster sampling is possible in the \cc{}.  For the hardcore model on graphs of constant degree, $\lambda = \frac{(\Delta-1)^{\Delta-1}}{(\Delta-2)^\Delta}$ serves as a computational hardness boundary in both the sequential and distributed settings.  In the sequential setting, there is a polynomial algorithm for sampling/counting when $\lambda$ is below the threshold by Weitz \cite{Weitz} and no polynomial algorithm for sampling/counting when $\lambda$ is above it unless RP=NP by Sly \cite{HardcoreThreshold}.  The same dynamic holds for sampling in the \local{} model, with an $O(\log^3 n)$ upper bound below the threshold \cite{distJVV} and an $\Omega(n^{1/11})$ lower bound above the threshold \cite{WhatCanBeSampled}.  Our results give an upper bound for $\lambda$ up to a constant fraction of this threshold.

\section{Technical Preliminaries}
\subsection{Congested Clique}
We take a very simple version of the \cc{} model where every vertex is represented by a machine with a unique ID of $O(\log n)$ bits.  Without loss of generality, we assume the vertices have IDs $0,1,2,...,n-1$.  Each machine holds its vertex's edges as well as its portion of the input.  For convenience we often mention the vertices themselves doing computation, referring to the machines holding the vertices.  The machines form an all-to-all communication network that communicates in synchronous rounds.  In each round, a machine can perform unbounded local communication and then send a message of size $O(\log n)$ bits to each other machine.  By the routing result of Lenzen \cite{LenzenRouting}, we can relax this requirement to each machine simply sending and receiving a total of $O(n)$ words of $O(\log n)$ bits each round, since these messages can all be successfully delivered within $O(1)$ deterministic rounds. 

\subsection{Gibbs Distributions on Graphical Models}
Consider a traditional graph labeling problem, where each vertex receives a label from some underlying alphabet $\mathcal{A}$.  We focus on randomly sampling from a subset, $\mathcal{S}$, of such labelings, given a distribution.  In particular, we want to sample from a Gibbs distribution.  Borrowing terminology from physics, we have the Hamiltonian, a function $H:\mathcal{S}\to \{0,1,...,h\}$, for some integer $h \geq 0$.  

\begin{definition}[Gibbs Distribution]
    The Gibbs distribution at inverse temperature $\beta\geq 0$ assigns a labeling, $\sigma\in\mathcal{S}$, a probability proportional to $\exp(-\beta H(\sigma))$.
\end{definition}  

In the trivial case, if we choose $H=0$, we get a uniform distribution over $\mathcal{S}$.  

\begin{definition}[Partition Function]
    The normalizing constant of the distribution, $Z(\beta) = \sum_{\sigma\in\mathcal{S}} \exp(-\beta H(\sigma))$, is referred to as the partition function.
\end{definition}  

We will make the assumption that $|\mathcal{A}|$ and $h$ are bounded by a polynomial in $n$.  This is not a significant restriction and is necessary both algorithmically and to ensure messages are of size $O(\log n)$.

For this paper we will focus heavily on two well studied distributions, the Potts model and the hardcore model.

\begin{definition}
    The Potts model is a Gibbs distribution where $\mathcal{A} = [q]$, is a set of colors.  $\mathcal{S}$ is $[q]^V$.  The Hamiltonian, $H$, assigns to each, possibly improper, coloring the number of monochromatic edges (edges where both vertices have the same color).  
\end{definition}

Note that when $\beta = \infty$, the Potts model simply becomes the uniform distribution over proper colorings.\footnote{We treat the infinite case in the sense of the limit as $\beta\to\infty$.}  Furthermore, $Z(\infty)$ counts the number of proper colorings.  On the other hand, when $\beta = 0$, the model is uniform over all proper and improper colorings.  Then $Z(0) = q^n$.  More precisely, when $\beta > 0$, this model is called the antiferromagnetic Potts model, since neighboring vertices favor not having the same color.

\begin{definition}
    The hardcore model is a distribution over the independent sets of a graph, parameterized by $\lambda\geq 0$ called the fugacity.  The probability of the independent set $I$ is proportional to $\lambda^{|I|}$.
\end{definition}
We can see that the hardcore model is also a Gibbs distribution by a change of variable.  Here $\mathcal{A} = \{0,1\}$.  We restrict $\mathcal{S}$ to require that the vertices labeled $1$ form an independent set.  The Hamiltonian, $H$, then counts the number of vertices given the label $1$.  We get the original distribution by choosing $\beta = -\ln\lambda$.

Here, $\lambda = 1, \beta = 0$, gives the uniform distribution over all independent sets.  Thus $Z(0)$ counts the number of independent sets.  Likewise, $\lambda = 0, \beta = \infty$, gives the distribution that only gives weight to the empty set.  Therefore, $Z(\infty) = 1$.  As a matter of convenience, the partition function is usually given in terms of $\lambda$ rather than $\beta$.  We will denote this $\hat{Z}(\lambda) = \sum_{\sigma\in\mathcal{S}}\lambda^{H(\sigma)}$.

\subsection{Counting From Sampling}
We would now like to estimate $Z(\beta)$, for a given choice of $\beta$.  We call this the generalized counting problem, as $Z$ gives the weighted count over all labelings.  In particular, either $Z(0)$ or $Z(\infty)$ typically give meaningful combinatorial counts.

Reductions from counting to sampling are well known.  Notably, Jerrum, Valiant, and Vazirani \cite{JVV}, showed that approximate sampling and randomized approximate counting are reducible to each other in polynomial-time.  A more efficient line of results have since developed for the specialized case of approximating the partition function of a Gibbs distribution, starting with Bezáková, Štefankovič, Vazirani, and Vigoda \cite{SimAnneal1} and then an improved result by Štefankovič, Vempala, and Vigoda \cite{SimAnneal2}.

We now discuss the high level technique of \cite{SimAnneal2}.  Our goal is to estimate $Z(\beta)$, for some $\beta > 0$ (often $\beta = \infty)$.  The basic idea, common to this area, is to rewrite $$Z(\beta) = Z(0) \times \frac{Z(\beta_1)}{Z(\beta_0 = 0)}\times\frac{Z(\beta_2)}{Z(\beta_1)}\times ... \times \frac{Z(\beta_l = \beta)}{Z(\beta_{\ell-1})}$$
for some sequence $\beta_0 = 0 < \beta_1 < ... < \beta_l = \beta$, called a cooling schedule.  We then estimate each term in the product to obtain the overall estimate of $Z(\beta)$.  We do require that $Z(0)$ is known in advance; however, this is usually simple to compute, for example in the Potts model it is $q^n$.  Note that in the case of the hardcore model, $Z(\infty)$ is known and not $Z(0)$, so we need to slightly rewrite the expression.  Further, as discussed earlier, $\hat{Z}(\lambda)$ is hard to approximate unless $\lambda \leq \frac{(\Delta-1)^{\Delta-1}}{(\Delta-2)^\Delta}$.

To estimate $\frac{Z(\beta_{i+1})}{Z(\beta_i)}$, we sample $\sigma$ from the Gibbs distribution at inverse temperature $\beta = \beta_i$.  We then compute $\exp((\beta_i - \beta_{i+1})H(\sigma))$ and note that it is an unbiased estimator for $\frac{Z(\beta_{i+1})}{Z(\beta_i)}$.  It follows from a result of Dyer and Frieze \cite{DyerFrieze} that the average of $O(\frac{\ell}{\epsilon^2})$ such estimators per term is sufficient to give an overall multiplicative $\epsilon$-estimation of $Z(\beta)$ if a simple condition is met on the cooling schedule.  In particular, for all $i$, $\frac{Z(\beta_i)Z(2\beta_{i+1}-\beta_i)}{Z(\beta_{i+1})^2} \leq B$, for some constant $B$.  In \cite{SimAnneal2}, it is shown how to find such a cooling schedule with length $O(\sqrt{\ln Z(0)}\ln\ln Z(0)\ln h)$.  Since we restrict $|\mathcal{A}|$ and $h$ to be $O(\text{poly}(n))$, this implies we need a total of $\Tilde{O}(\frac{n}{\epsilon^2})$ samples for an approximate count.  The problem with running this approach in parallel is that the cooling schedule can only be computed dynamically.  In other words, not all of these samples can be taken at a time.

\subsection{Parallel Counting From Sampling}
A more complex product estimation scheme for the partition function, the paired product estimator, was found by Huber \cite{HPPE}.  Very recently, Liu, Yin, and Zhang \cite{ParGibCount}, showed how to combine the two estimators to allow for parallel sampling.  Again, we need a total of $\Tilde{O}(\frac{n}{\epsilon^2})$ samples to find the count.  Unlike the approach of \cite{SimAnneal2}, the samples can all be taken simultaneously rather than being dependent on each other.  The following theorem summarizes their result when $h$ and $|\mathcal{A}|$ are $O(\text{poly}(n))$.

\begin{theorem}\label{simulatedAnnealingTheorem}
    For a Gibbs distribution and inverse temperatures $0\leq\beta_1<\beta_2$, $\frac{Z(\beta_2)}{Z(\beta_1)}$ can be computed with a multiplicative error of $\epsilon$ with probability at least $\frac{3}{4}$ given the Hamiltonian of $\Tilde{O}(\frac{n}{\epsilon^2})$ samples taken from the Gibbs distribution for inverse temperatures ranging from $[\beta_1,\beta_2]$.  Furthermore, the temperature each sample is taken at can be precomputed and only depends on $n$, $h$, and $|\mathcal{A}|$.
\end{theorem}

The theorem implies that we can approximate $Z(\beta)$, as long as $Z(\beta')$ is known in advance for some $\beta'\geq 0$ and that we can draw $\Tilde{O}(\frac{n}{\epsilon^2})$ samples for inverse temperatures in the range $[\min(\beta,\beta'), \max(\beta,\beta')]$.

\subsection{Pairwise Markov Random Fields}
The traditional way to sample from a Gibbs distribution is using a rapidly mixing Markov chain.  We begin from an arbitrary state (a graph labeling) and simulate transitions until the chain is \textit{mixed}, meaning that the current state is roughly drawn from the desired distribution.  Here we will focus on sampling from the probability space defined by a pairwise Markov random field.  A pairwise Markov random field on a graph consists of vertex constraints and edge constraints.  

For each vertex $v$, we have a constraint $b_v:\mathcal{A}\to\mathbb{R}_{\geq 0}$.  For each edge $e = (u,v)$ we also have a constraint $A_e:\mathcal{A}\times\mathcal{A}\to\mathbb{R}_{\geq 0}$.  Here the edges are undirected, but the function $A_e$ is not necessarily symmetric.  The probability of graph vertex labeling $l$ is defined to be proportional to $$\prod_v b_v(l(v)) \cdot \prod_{e = (u,v)}A_e(l(u), l(v))$$

\begin{example}[Hardcore Model]$ $\\
Choosing $b_v(x) = \begin{cases}
    1,\, \text{if } x = 0\\\lambda, \text{if } x=1
\end{cases}$
and $A_e(x,y) = \begin{cases}
    0,\, \text{if } x = y = 1\\1, \text{else}
\end{cases}$
we get the hardcore model.
\end{example}

Similarly, we can get the Potts model subject to the same change of variable, $\beta = -\ln\lambda$.

\begin{example}[Potts Model]We choose $b_v(x) = 1$
and $A_e(x,y) = \begin{cases}
    \lambda,\, \text{if } x = y\\1, \text{else}
\end{cases}$\end{example}

Finally, this leads us to our definition of a \textit{local Gibbs distribution}.

\begin{definition}[Local Gibbs Distribution]We call a Gibbs distribution local if
\begin{enumerate}
    \item The distribution can be represented as a pairwise Markov random field.
    \item The Hamiltonian $H$ can be represented as a sum over sub-functions defined on each edge and vertex.  $H(\sigma) = \sum_v H_v(\sigma) + \sum_{e=(u,v)}H_e(u,v)$.  Furthermore, each of these subfunctions also has a codomain of $[h]$.
    \item An element $\sigma\in\mathcal{S}$ can be constructed in the \cc{} in $O(\log n)$ rounds.
\end{enumerate}

Clearly the Potts and hardcore models are both examples of local Gibbs distribution.  The third condition is needed to initialize the Markov chain, and it will often hold by a variant of Luby's algorithm \cite{Lubys}.
\end{definition}

\subsection{Distributed Sampling}
Feng, Sun, and Yin \cite{WhatCanBeSampled} found a preliminary distributed Markov chain to sample from pairwise Markov random fields.  In the case of $q$-colorings, this was improved simultaneously by Fischer and Ghaffari \cite{FischerGhaffari} and Feng, Hayes, and Yin \cite{FengHayesYin}, by having vertices only become active with probability $p$.  This symmetry breaking step improved the range of $q$ over which colorings could be quickly sampled.  In particular, their improved chain samples $q$-colorings within a total variation distance of $\delta$ in $O(\log \frac{n}{\delta})$ \congest{} rounds when $q \geq \alpha\Delta$ for $\alpha > 2$.  The improved chain can be generalized beyond colorings, see \cite{ExactDistSampl} by Pemmaraju and Sobel for the chain explicitly given in full generality.  One transition of the chain, for some yet to be determined probability $p$, is defined in the following pseudocode.  Every vertex $v$ has its current label $X_v$.  For ease of notation, in the edge acceptance step, any fraction that is not well defined (because certain vertices may not be active) equals $1$.  Note that on some pairwise Markov random fields this chain may have a very large mixing time, $\gg\log n$.  In fact, there are examples such as $\Delta+1$ coloring, where it is not even possible to transition in finite steps between two states of positive probability.  In these cases the mixing time is infinite.

\begin{algorithm}[H]
\begin{algorithmic}
\caption{Distributed Metropolis-Hastings}
\For{each vertex $v$}\State $v$ becomes active with probability $p$
\If{$v$ is active}\State $v$ samples a proposed color, $\sigma_v$, from a distribution proportional to $b_v(\sigma)$\EndIf
\EndFor

\For{each edge $e = (u,v)$}\State edge $e$ is accepted with probability $\frac{A_e(X_u, \sigma_v)}{\max A_e} \cdot\frac{A_e(\sigma_u, X_v)}{\max A_e} \cdot\frac{A_e(\sigma_u, \sigma_v)}{\max A_e}$
\EndFor

\For{each active vertex $v$}\State $X_v \gets \sigma_v$, if every edge incident to $v$ is accepted
\EndFor
\end{algorithmic}
\end{algorithm}

\section{Our Results}
\subsection{Fast Mixing for the Hardcore and Potts Models}
First, we show that the distributed Metropolis-Hastings Markov chain is fast mixing for certain parameters of the hardcore and Potts models.  In particular, the chain will return a state sampled with total variation distance error at most $\delta$ after $O(\log\frac{n}{\delta})$ transitions.

\begin{restatable}{theorem}{hardcoremixing}
    The distributed Metropolis-Hastings chain mixes in $O(\log \frac{n}{\delta})$ transitions for the hardcore model when $\lambda \leq \frac{\alpha}{\Delta-1}$ for $\alpha < 1$, when $\Delta \geq 2$.
\end{restatable}
\begin{proof}
We will use the standard coupling approach.  Consider two instances of the Markov chain, $X$ and $Y$, that only differ at a single vertex $v$.  In particular $X_v\neq Y_v$.  Now imagine that in the next step of the chain the same vertices are marked active in both chains with the same proposals.  It is sufficient to show that the expected Hamming distance between $X$ and $Y$ after this step is at most a constant less than $1$.

Note that the only vertices that could possibly differ are the vertices in the inclusive neighborhood of $v$.  First consider vertex $v$ itself.  If $v$ is not marked active it will definitely differ between the two chains.  Otherwise, it will differ if it has an active neighbor that is proposing $1$.  Taking the union bound across all such events, we get that the probability $v$ differs between the two chains is at most $1-p+p(\Delta p \frac{\lambda}{1+\lambda})$.  Now, consider a neighbor of $v$, $w$.  $w$ may differ in the two chains if it is active and proposing $1$.  Altogether, taking the union bound over all neighbors of v, we get that expected Hamming distance is bounded by
$$1-p+p\bigg(\Delta p \frac{\lambda}{1+\lambda}\bigg)+\Delta p \frac{\lambda}{1+\lambda}$$
Simplifying, we see that for any choice of $0 \leq \alpha < 1$, there exists a choice of $p$, such that the expression is bounded by a constant less than $1$.
\end{proof}

We use the coupling from \cite{FengHayesYin}, with slight modification, to generalize their fast mixing time result for $\alpha\Delta$-proper colorings to the Potts model.  This extension is not surprising, since intuitively mixing will take longest when $\beta = \infty$, the uniform distribution over proper colorings.  As $\beta\to 0$, edges are more likely to be accepted leading to faster mixing.  Since the proof closely follows the structure of Theorem 1 in \cite{FengHayesYin}, we defer it to the appendix.

\begin{restatable}{theorem}{mixingpotts}The distributed Metropolis-Hastings chain mixes in $O(\log\frac{n}{\delta})$ transitions for the Potts model for $\beta\geq 0$, when $q \geq \alpha\Delta$, for $\alpha > 2$.
\end{restatable}
\subsubsection{Quick Note About Sampling Precision}
Since we are using approximate rather than exact sampling in our algorithm, it raises the question of what total variation distance error, $\delta$, our sampler should allow.  By a standard coupling argument, if we take a total of $k$ samples with error $\delta$, we increase the probability of the algorithm failing by at most $k\delta$.  If $\delta = 0$, we get that the algorithm fails with probability at most $\frac{1}{4}$.  Thus it is sufficient for $k\delta \leq \frac{1}{8}$, so that the total algorithm failure probability is bounded by a constant less than $\frac{1}{2}$.  This can be boosted with the median trick to give a success probability at least $\frac{3}{4}$.  Finally, we assume $\epsilon > \frac{1}{\sqrt{n}}$, otherwise our result is trivial.  This means that from now on we can always choose $\delta = \frac{1}{8k}$, and an $O(\log\frac{n}{\delta})$ sampling algorithm simply takes $O(\log n)$ rounds.
\subsection{Fast Simulation of Distributed Metropolis-Hastings}
We would now like to estimate the partition function of a local Gibbs distribution.  As described earlier, we need to know the Hamiltonian of $k = \Tilde{O}(\frac{n}{\epsilon^2})$ samples from the Gibbs distribution.  This can be found by running $k$ instances of the distributed Metropolis-Hastings Markov chain.  We can simulate one transition of $n$ instances of the chain in parallel.

\mainresult*

\subsubsection{Algorithm}
\paragraph*{Initialization}Each chain starts at an initial labeling.  At time step $t$, each vertex will hold $X_{v,i,t}$, the label of vertex $v$ in chain $i$.  It will also sample its new proposed label, $\sigma_{v,i,t}$.  In the case the vertex is inactive in a chain, $\sigma_{v,i,t} = \sqcup$.

To decide if $X_{v,i,t+1} \gets \sigma_{v,i,t}$, assuming $\sigma_{v,i,t}\neq\sqcup$, $v$ needs to know if any of its incident edges are rejected in chain $i$.  We will use the technique from Dolev, Lenzen, and Peled \cite{TriangleCounting} for triangle counting and also used by Censor-Hillel, Kaski, Korhonen, Lenzen, Paz, Suomela \cite{MatrixMult} for matrix multiplication over semirings.

For convenience, we will use $0$ based indexing.  Let $A$ be the adjacency matrix of the graph.  Consider an $n\times n\times n$ cube, $Q$, where $Q_{v,w,i}$ is true iff $A_{v,w} = 1$ and $\{v,w\}$ is not accepted in chain $i$.  Furthermore, think of $Q$ evenly divided into $n$, $n^{2/3}\times n^{2/3}\times n^{2/3}$ subcubes.  Let subcube $Q[x,y,z]$ refer to the subcube over the indices $[xn^{2/3}, (x+1)n^{2/3})\times [yn^{2/3}, (y+1)n^{2/3}) \times [zn^{2/3}, (z+1)n^{2/3})$.  We will use a similar notation for matrices, $A[x,y]$ refers to the submatrix over the indices $[xn^{2/3}, (x+1)n^{2/3})\times [yn^{2/3}, (y+1)n^{2/3})$.  

\paragraph*{Computing $Q$}Each subcube is arbitrarily assigned to a machine that will be responsible for generating the entries of the subcube.  Because edges are non-directional, we assign the subcubes $Q[x,y,z]$ and $Q[y,x,z]$ to the same machine.  Each machine is assigned at most two subcubes.
    
Consider the machine $\sf{M}$ holding the subcubes $Q[x,y,z]$ and $Q[y,x,z]$.  To fill in the subcubes, $\sf{M}$ needs to hold $A[x,y]$, $\sigma[x,z]$, $\sigma[y,z]$, $X[x,z]$, and $X[y,z]$.  Altogether, $\sf{M}$ needs to receive $\Theta(n^{4/3})$ words.  For every machine to do this step in parallel, each machine needs to send and receive $\Theta(n^{4/3})$ words.  This communication is done in $\Theta(n^{1/3})$ rounds, using Lenzen's result \cite{LenzenRouting}.  $\sf{M}$ then uses the information it received in addition to its randomness to fill in the (possibly two) subcube(s) it is responsible for.  More precisely, $\sf{M}$ runs the edge acceptance step of the Markov chain for the edge-chain pairs it is responsible for and fills in the subcube(s).

\paragraph*{Updating Labels}
$\sf{M}$ will compress each subcube it holds into a submatrix by applying union along the $y$-axis.  In particular, the subcube $Q[x,y,z]$ is flattened into the submatrix $R^y[x,z]$, where $R^y[x,z]_{i,k} = \text{true}$ if there exists $j$ such that $Q[x,y,z]_{i,j,k} = \text{true}$.

Finally, vertex $u$ accepts its proposal in chain $i$ if $\lor_y R^y_{u,i} = \text{false}$, as this indicates that no incident edge to $u$ was rejected.  To compute this for every chain, $u$ receives $R^y_{u,i}$ for all $y$ and $i$.  Doing this simultaneously for all vertices, requires each machine to send and receive $O(n^{4/3})$ words of information.  This takes a final $O(n^{1/3})$ rounds.

\begin{corollary}With success probability at least $\frac{3}{4}$, we can approximate the partition function, $Z(\beta)$, of a local Gibbs distribution with multiplicative error $\epsilon$ in $\Tilde{O}(\frac{n^{1/3}}{\epsilon^2})$ rounds, when $Z(\beta')$ is known for some $\beta'$ and the associated distributed Metropolis-Hastings chain has an $O(\log n)$ mixing time for inverse temperatures in the interval $[\min(\beta,\beta'), \max(\beta,\beta')]$.
\end{corollary}
\begin{proof}
    \cref{simulatedAnnealingTheorem} and \cref{simulationTheorem} almost immediately give us the result.  We draw $\Tilde{O}(\frac{n}{\epsilon^2})$ samples from the Gibbs distribution by simulating $O(\log n)$ transitions of the Markov chain.  There is one slight problem, the samples are held in a distributed fashion and we need the Hamiltonian of each sample held centrally at some machine.

    We use a very similar strategy as before, with a slight change.  We now let $Q_{v,w,i}$ denote $H_{(v,w)}(X_{v,i},X_{w,i})$.  We use the same subcube division scheme and compute this for every subcube.  However, now when we compress a subcube into a matrix we use addition instead of union.  Then, the sum of edge Hamiltonian subfunctions incident to $u$ in chain $i$, is $\sum_y R^y_{u,i}$.  
    For each chain $i$, $u$ computes its portion of the overall Hamiltonian, $H^{u,i} = H_u(X_{u,i}) + \frac{1}{2}\sum_y R^y_{u,i}$.  Then $u$ can send $H^{u,i}$ to machine $i$.  Machine $i$ can compute the entire Hamiltonian for chain $i$, $\sum_{u}H^{u,i}$.  Lastly, all such Hamiltonian values can be gathered at a singular central machine.
\end{proof}

\colorcount*

In the special case of the hardcore model, we can do better than the general simulation result.  
\begin{theorem}We can draw $n$ samples from the distributed Metropolis-Hastings chain for the hardcore model in $\Tilde{\Theta}(1)$ rounds, when $\lambda \leq \frac{\alpha}{\Delta-1}$ for $\alpha < 1$, with probability at least $1-\frac{1}{n^2}$.
\end{theorem}
\begin{proof}
   For the hardcore model, we can take advantage of the fact that there are only two possible labels.  Imagine each chain starting from the empty independent set.  For one vertex to update its label in a transition, it only needs to know which of its neighbors have the current label $1$ and which of its neighbors are proposing the label $1$.

   Consider some vertex $v$.  We will first bound the total messages $v$ needs to send.  Let $X_{v,i}$ indicate whether $v$ ever proposes $1$ in chain $i$ up to the mixing time.  Note that $E[X_{v,i}]=\Tilde{O}(\frac{1}{\Delta})$.  Note that the total number of messages $v$ ever needs to send each of its neighbors is bounded by $\Tilde{O}(\sum_i X_{v,i})$.  By the Chernoff bound, this is $\Tilde{O}(n\frac{1}{\Delta})$ with probability at least $1-\frac{1}{2n^3}$.  Since these messages need to be sent to all of v's neighbors, we get that $v$ needs to send a total of $\Tilde{O}(n)$ messages with probability at least $1-\frac{1}{2n^3}$.  Taking the union bound over all vertices, no vertex sends more than $\Tilde{O}(n)$ messages with probability at least $1-\frac{1}{2n^2}$.

   Now we will bound the number of messages $v$ ever receives.  Similarly, this is bounded by $\Tilde{O}(\sum_i\sum_{u\in N(v)} X_{u,i})$.  Applying the Chernoff and union bounds again, no vertex sends more than $\Tilde{O}(n)$ messages with probability at least $1-\frac{1}{2n^2}$.

   Altogether, with probability at least $1-\frac{1}{n^2}$, no vertex ever sends or receives more than $\Tilde{O}(n)$ messages.  By the result of Lenzen \cite{LenzenRouting}, each transition of all $n$ chains can be simulated in $\Tilde{O}(1)$ rounds.
\end{proof}

\hardcorecounting*

\subsection{Lower Bound}
It is not surprising that conditional lower bounds exist for estimating the partition function of a local Gibbs distribution in the \cc{}, since for many distributions, such as the hardcore model with high fugacity, it is known to be NP-hard.  However, we can show that even in the case where the distributed Metropolis-Hastings chain has mixing time $O(\log n)$, we still have conditional hardness results.

Consider the following Gibbs distribution.  We choose $\mathcal{A} = V\cup [3n]$.  We restrict the set $\mathcal{S}$ to be all labelings such that every vertex's label is either in $[3n]$ or is a neighboring vertex.  We let the Hamiltonian count the number of edges $\{u,v\}$ such that the labels of $u$ and $v$ are the same vertex.  It is not hard to show that this is a local Gibbs distribution and that the mixing time of the distributed Metropolis-Hastings chain is $O(\log\frac{n}{\delta})$ for $\beta\geq 0$.  The mixing time easily follows because the chain has even more degrees of freedom than for $q = 3\Delta$ colorings.

As always, when $\beta = 0$, all labelings in $\mathcal{S}$ are equally likely.  $Z(0) = \prod_v 3n+\deg(v)$.  On the other hand, when $\beta=\infty$, we remove labelings where two neighboring vertices are labeled the same mutually neighboring vertex.  This can only happen if there is a triangle in the graph.  Thus, $Z(0) > Z(\infty)$ if and only if there is a triangle in the graph.  In fact, if the graph has at least one triangle, $a,b,c$, then $1 - \frac{Z(\infty)}{Z(0)} \geq \frac{1}{16n^2}$, since that is a lower bound on the probability that vertices $a$ and $b$ both choose $c$ in a uniform random labeling.  Rearranging, $Z(\infty) \leq Z(0) (1 - \frac{1}{16n^2}) \leq Z(0) (1 - \frac{1}{32n^2})^2$.

Note that we can easily compute $Z(0)$ exactly.  Now we can choose $\epsilon = \frac{1}{32n^2}$ and take an approximate count of $Z(\infty)$, $\bar{Z}$.  We have two cases to consider.  In the first, the graph doesn't have a triangle.  Then, $Z(\infty) = Z(0)$.  Thus, $\bar{Z} \geq (1-\frac{1}{32n^2})Z(0)$.  In the other case, the graph has a triangle.  Then, $\bar{Z} \leq (1+\frac{1}{32n^2})Z(\infty)\leq (1 - \frac{1}{32n^2})^2(1+\frac{1}{32n^2})Z(0) < (1-\frac{1}{32n^2})Z(0)$.  Thus, we can use $\bar{Z}$ to determine if the graph has a triangle.  This example leads to the following theorem.

\begin{theorem}\label{lowerBound}Approximating the partition function of a local Gibbs distribution when $\epsilon \leq \frac{1}{32n^2}$ with success probability $p$ is as hard as triangle detection with success probability $p$.  This holds even when $Z(0)$ is known and the distributed Metropolis-Hastings chain is fast mixing over all $\beta\geq 0$.
\end{theorem}

Note that currently the best algorithm for triangle detection uses matrix multiplication and takes $O(n^{0.157})$ rounds \cite{MatrixMult,NewMatrixMult}.

\section{Conclusion}
We present a fast algorithm for approximating the partition function of a local Gibbs distribution in the \cc{}, by utilizing the parallel counting to sampling reduction of \cite{ParGibCount}.  In particular, their result requires $\Tilde{O}(\frac{n}{\epsilon^2})$ samples to obtain a count.  We can take $n$ samples in parallel in the \cc{} by observing an interesting similarity to triangle counting and semiring matrix multiplication\cite{MatrixMult,TriangleCounting}.  For some models, such as the hardcore model at low fugacity, we can do even better by bounding the total number of messages each vertex needs to send.  We also obtain the conditional lower bound that approximating the partition function of a local Gibbs distribution is at least as hard as triangle detection when $\epsilon\leq\frac{1}{32n^2}$.  The gap between this lower bound and our result suggests two possible areas of improvement: reducing the runtime dependence of our algorithm on $\epsilon$ and reducing the runtime needed to take $n$ samples.  Reducing the runtime dependence on $\epsilon$ seems difficult because even the state of the art sequential algorithms require $\frac{n}{\epsilon^2}$ samples.  On the other hand, due to structural similarities between the problems, taking $n$ samples in integer matrix multiplication time may be possible, however this remains unclear when integers are limited to $O(\log n)$ bits.

\bibliography{refs}
\appendix
\section{Fast Mixing for the Potts Model}
\mixingpotts*
\begin{proof}
For convenience, we will make the change of variable $\lambda = \exp(-\beta)$.  We use the same coupling from \cite{FengHayesYin}, with slight modification.  We first take a slightly different view of the Markov chain.  While the edges in the input graph are undirected, here we think of edge $\{u,v\}$ being replaced by two directed edges, $(u,v)$ and $(v,u)$.  Instead of having vertices first make proposals and then determining which edges are accepted, we will have each pair of edges, $(u,v)$ and $(v,u)$, independently flip three coins before the proposal process, where the probability of heads of each coin is $\lambda$.  In particular, the edges $(u,v)$ and $(v,u)$ flip the same 3 coins, but coin 1 for $(u,v)$ is coin 2 for $(v,u)$ and vice versa.  Coin 3 is the same for both $(u,v)$ and $(v,u)$.  The edge $(u,v)$ is then accepted after proposals are made if each of the following three conditions are met (for proposals that are null we consider the corresponding conditions to be met).
\begin{itemize}
    \item $X_u \neq \sigma_v$ or coin 1 is heads
    \item $X_v \neq \sigma_u$ or coin 2 is heads
    \item $\sigma_u \neq \sigma_v$ or coin 3 is heads
\end{itemize}
Note that $(u,v)$ and $(v,u)$ are either both accepted or both rejected.  Finally, like before, a vertex accepts its proposal if every incident edge is accepted.  It's easy to see this process is the same as the distributed Metropolis-Hastings chain for the Potts model.

We now define type 1 active edges to mean that coin 1 is tails.  We define type 2 active edges to mean that coin 3 is tails.  Here, $\Tilde{N}(v)$, is the neighborhood of $v$ along type 1 active edges, $\{w | (v,w) \text{ is a type 1 active edge}\}$.

Consider the following coupling of chains $X$ and $Y$ that differ at a single vertex $v'$, $X_{v'}$ is red and $Y_{v'}$ is blue.  We first flip the same coins in both chains.  The same vertices are then marked active in both chains.  The active vertices of $X$ then choose their proposed colors uniformly at random.  Now all active vertices that are not in $\Tilde{N}(v')$ propose the same colors in chain $Y$.  A vertex $u \in \Tilde{N}(v')$ proposes the same color as in chain $X$ except that the colors blue and red are swapped, unless it neighbors a vertex $w$ other than $v'$ and either
\begin{itemize}
    \item $(w,u)$ is a type 1 active edge and $X_w$ is red/blue
    \item $w\notin\Tilde{N}(v')$, $(w,u)$ is a type 2 active edge, and $\sigma_w$ is red/blue
\end{itemize}
then it proposes the same color as in $X$.  Now we analyze the coupling.

First consider a vertex $v$ in $V/(\Tilde{N}(v')\cup \{v'\})$.  $X_v = Y_v$ and assuming it is active $\sigma^X_v = \sigma^Y_v$.  Thus the chains can only disagree at this vertex after the transition if an incident edge $e = \{v,w\}$ is accepted in one chain and not the other.  This leads to three cases: $(v,w)$ is type 1 active and fails the first condition in exactly one chain, $(w,v)$ is type 1 active and fails the first condition in exactly one chain, or $e$ is type 2 active and fails the third condition in exactly one chain.  Now in the first case we must have $\sigma^X_w \neq \sigma^Y_w$ and further that either $\sigma^X_w = X_v$ or $\sigma^Y_w = X_v$.  Since $w's$ proposals can only differ if they are red or blue, this would imply $X_v$ is red or blue.  But this is impossible since $w$ can't have a flipped proposal if $(v,w)$ is type 1 active and $X_v$ is blue or red.  Now consider the second case, the edge $(w,v)$ is type $1$ active.  We know $X_w=Y_w$ since $v\notin\Tilde{N}(v')$.  Furthermore, $\sigma^X_v = \sigma^Y_v$.  Thus the condition either passes in both chains or fails in both chains.  Finally, we consider the case where $e$ is type 2 active.  This is impossible for a very similar reason as the first case.  Altogether, after the transition, vertex $v$ must be the same in both chains.

Next we consider a vertex $v$ in $\Tilde{N}(v')$.  Note that $v$ can only differ after the transition when it proposes red/blue.  Furthermore, closer analysis shows that regardless of whether $v$ is proposing consistently or flipped in the two chains, there is actually only a single color it can choose to cause a difference at the transition.  Thus the probability for each such vertex is at most $\frac{p}{q}$.

Finally, we consider the vertex $v'$.  We can say that it will not differ in the chains after the transition if several conditions hold.  First it is marked active.  Second, it proposes a color distinct from all of its neighbors current colors.  Third, each neighbor is either inactive or proposes a color that is not red/blue and not what $v'$ is proposing.  This all happens with probability at least $p(\frac{q-\Delta}{q})(1-p+p\frac{q-3}{q})^\Delta$

Altogether, the expected Hamming distance between the chains after the transition can be bounded by taking the sum over all vertices of the probabilities we just computed.  This gives an upper bound of $1-p(\frac{q-\Delta}{q})(1-p+p\frac{q-3}{q})^\Delta + \frac{\Delta p}{q}$, the exact same bound achieved in \cite{FengHayesYin}.  They show that for $\alpha>2$ and sufficiently small $p$ this is less than $1$, completing the proof.
\end{proof}
\end{document}